\documentclass[11pt,english,JHEP]{article}
\usepackage{ae,aecompl}

\usepackage[T1]{fontenc}
\usepackage[latin9]{inputenc}
\usepackage{color}
\usepackage{babel}
\usepackage{amsmath}
\usepackage{amsthm}
\usepackage{amssymb}
\usepackage{setspace}
\usepackage[unicode=true,pdfusetitle,
 bookmarks=true,bookmarksnumbered=false,bookmarksopen=false,
 breaklinks=true,pdfborder={0 0 0},pdfborderstyle={},backref=false,colorlinks=true]
 {hyperref}
\hypersetup{
 urlcolor=magenta, citecolor=blue, linkcolor=red}

\makeatletter
\theoremstyle{plain}
\newtheorem{thm}{\protect\theoremname}
\theoremstyle{remark}
\newtheorem{rem}[thm]{\protect\remarkname}
\theoremstyle{plain}
\newtheorem{lem}[thm]{\protect\lemmaname}
\theoremstyle{definition}
\newtheorem{example}[thm]{\protect\examplename}

\definecolor{green}{RGB}{100, 0, 120}

\newlength{\abstractwidth}
\g@addto@macro\normalsize{%
 \setlength{\abovedisplayskip}{10pt}
 \setlength{\belowdisplayskip}{14pt}
 \setlength{\abovedisplayshortskip}{10pt}
 \setlength{\belowdisplayshortskip}{14pt}}

\usepackage{color}

\makeatother

\providecommand{\examplename}{Example}
\providecommand{\lemmaname}{Lemma}
\providecommand{\remarkname}{Remark}
\providecommand{\theoremname}{Theorem}

\begin{document}

\title{\textbf{Entanglement cost of discriminating noisy Bell states by
local operations and classical communication }}

\author{Somshubhro Bandyopadhyay\thanks{\texttt{Department of Physics, Bose Institute, Kolkata 700091, India;}
Email: \texttt{som@jcbose.ac.in}}$\hspace{4em}$Vincent Russo\thanks{\texttt{ISARA Corporation, Waterloo, ON N2L0A9, Canada}; Email: \texttt{vincentrusso1@gmail.com}{\large{} }}}
\maketitle
\begin{abstract}
Entangled states can help in quantum state discrimination by local
operations and classical communication (LOCC). For example, a Bell
state is necessary (and sufficient) to perfectly discriminate a set
of either three or four Bell states by LOCC. In this paper, we consider
the task of LOCC discrimination of the states of noisy Bell ensembles,
where a given ensemble consists of the states obtained by mixing the
Bell states with an arbitrary two-qubit state with nonzero probabilities.
It is proved that a Bell state is required for optimal discrimination
by LOCC, even though the ensembles do not contain, in general, any
maximally entangled state, and in specific instances, any entangled
state.
\end{abstract}

\section{Introduction }

The paradigm of $local$ $operations$ $and$ $classical$ $communication$
(LOCC) is of particular importance in quantum information theory \cite{LOCC}.
LOCC protocols involve two or more parties sharing a composite quantum
system who perform arbitrary quantum operations on the local subsystems
and communicate only via classical channels. Note that quantum communication
is not allowed between the parties. The framework of LOCC provides
a natural way to study the resource theory of quantum entanglement
\cite{Entanglement-horodecki}, the nonlocal properties of quantum
systems \cite{Brunner-et-al-2014,Gisin-1996,Camlet-2017,Peres-Wootters-1991,Massar-Popescu-1995,ben99,walgate-2002,HSSH,nis06,Bandyo-2011,Halder+-2019},
and applications thereof \cite{Barrett+2005,Acin+2006,Brunner+2008,Pironio+2010,Terhal2001,DiVincenzo2002,Eggeling2002,MatthewsWehnerWinter,Markham-Sanders-2008}. 

\subsection*{Local state discrimination}

One problem that has been extensively studied within the framework
of LOCC is discrimination of quantum states \cite{ben99,walgate-2002,HSSH,Bandyo-2011,Halder+-2019,ben99u,Walgate-2000,Virmani-2001,Ghosh-2001,divin03,rin04,Ghosh-2004,fan-2005,Watrous-2005,Nathanson-2005,Wootters-2006,Hayashi-etal-2006,Duan2007,Duan-2009,feng09,Calsamiglia-2010,BGK-2011,Yu-Duan-2012,Cosentino-2013,Cosentino-Russo-2014,halder}.
The problem may be briefly described as follows. Let $\mathcal{E}=\left\{ \left(p_{i},\rho_{i}\right):i=1,\dots,N\right\} $
be an ensemble of $k$-partite quantum states $\rho_{1},\dots,\rho_{N}$
with associated probabilities $p_{1},\dots,p_{N}$, where $k,N\geq2$.
Now suppose that $k$ separated parties share a quantum system prepared
in a state chosen from $\mathcal{E}$. The parties do not know the
identity of the state but they do know that the state has been chosen
from $\mathcal{E}.$ The goal is to gain as much knowledge about the
state of the system by means of LOCC. For example, if the given states
are mutually orthogonal, then they wish to find $which$ $state$
the system is in without error. This is evidently a state discrimination
problem in which the allowed measurements are only those that are
implementable by LOCC. So the question of interest here is the following:
For a given set of states, does there exist a LOCC measurement that
discriminates the states just as well as the best possible measurement
that may be performed on the whole system? 

In general, how well a given set of states can be discriminated can
be quantified by the success probability for minimum-error state discrimination\footnote{For general discussions on minimum-error discrimination one may consult
\cite{Holevo-1973,Helstrom-1976,Eldar+2003,Qiu-2008,Qiu-Li-2010}. }. The success probability is the optimized value of the average probability
of success, where the optimization is over either all measurements
or some specific class of measurements. Thus, for a given ensemble
$\mathcal{E}$, let $p\left(\mathcal{E}\right)$ and $p_{_{{\rm L}}}\left(\mathcal{E}\right)$
\cite{BN-2013} denote the success probability (global optimum) and
the local success probability (local optimum), where the corresponding
optimizations are taken over all measurements and LOCC measurements,
respectively. 

Let us now come back to the question of whether the global optimum
for a given set of states is always achievable by LOCC. The answer
turns out to be no in general; that is, sets that cannot be optimally
discriminated by LOCC, even if the states are all pure and mutually
orthogonal, exist. Once the initial results \cite{ben99,ben99u} established
this fact, most of the subsequent works were devoted to identifying
and characterizing the sets for which the global optimum is achievable
by LOCC \cite{Walgate-2000,Virmani-2001,Nathanson-2005} and those
for which it is not (e.g., \cite{walgate-2002,Ghosh-2001,Ghosh-2004,fan-2005,Watrous-2005,Nathanson-2005}).
For example, two pure states can be optimally discriminated by LOCC
\cite{Walgate-2000,Virmani-2001} but an entangled orthogonal basis,
such as the Bell basis, cannot be \cite{HSSH,Ghosh-2001,Ghosh-2004,Nathanson-2005}.
So given a set of states that cannot be optimally discriminated by
LOCC, one must, therefore, consider using quantum entanglement as
a resource for optimal discrimination. 

\subsection*{Entanglement as a resource for local state discrimination}

The limitations of LOCC protocols in discriminating quantum states
can be overcome with shared entanglement used as a resource \cite{B-IQC-2015,BGKW-2009,BRW-2010,Cohen-2008,BHN-2016,BHN-2018,LJ-2020}.
Consider a simple example: The Bell basis $\mathcal{B}$, which is
defined by the four Bell states, 
\begin{equation}
\begin{array}{cccc}
\left|\Psi_{1}\right\rangle =\frac{1}{\sqrt{2}}\left(\left|00\right\rangle +\left|11\right\rangle \right), &  &  & \left|\Psi_{2}\right\rangle =\frac{1}{\sqrt{2}}\left(\left|00\right\rangle -\left|11\right\rangle \right),\\
\left|\Psi_{3}\right\rangle =\frac{1}{\sqrt{2}}\left(\left|01\right\rangle +\left|10\right\rangle \right), &  &  & \left|\Psi_{4}\right\rangle =\frac{1}{\sqrt{2}}\left(\left|01\right\rangle -\left|10\right\rangle \right),
\end{array}\label{Bell-basis}
\end{equation}
cannot be perfectly discriminated by a LOCC measurement, even though
the states are mutually orthogonal \cite{Ghosh-2001}. In particular,
assuming the states are equally probable, one has (see e.g., \cite{BN-2013})
\begin{equation}
p_{_{{\rm L}}}\left(\mathcal{B}\right)=\frac{1}{2}.\label{p-L-B}
\end{equation}
Now suppose the parties are given a two-qubit ancillary state\footnote{One can assume the given form of $\left|\tau_{\varepsilon}\right\rangle $
because of the Schmidt decomposition.} 
\begin{equation}
\left|\tau_{\varepsilon}\right\rangle =\sqrt{\frac{1+\varepsilon}{2}}\left|00\right\rangle +\sqrt{\frac{1-\varepsilon}{2}}\left|11\right\rangle \label{tau}
\end{equation}
for some $\varepsilon\in\left[0,1\right]$, where $\left|\tau_{\varepsilon}\right\rangle $
is entangled for $0\leq\varepsilon<1$. Once again assuming the Bell
states are equiprobable, the local success probability for discriminating
the set of states
\[
\mathcal{B}\otimes\tau_{\varepsilon}=\left\{ \left|\Psi_{i}\right\rangle \otimes\left|\tau_{\varepsilon}\right\rangle :i=1,\dots,4\right\} 
\]
is given by \cite{B-IQC-2015} 
\begin{equation}
p_{_{{\rm L}}}\left(\mathcal{B}\otimes\tau_{\varepsilon}\right)=\frac{1}{2}\left(1+\sqrt{1-\varepsilon^{2}}\right)\label{p-L-Btimes-tau}
\end{equation}
for all $\varepsilon\in\left[0,1\right]$. This value is achievable
by a teleportation protocol\footnote{The teleportation protocol is the following: Alice teleports her part
of the unknown two-qubit state to Bob using $\left|\tau_{\varepsilon}\right\rangle $.
After teleportation, Bob performs a two-qubit measurement to discriminate
the states. If $\left|\tau_{\varepsilon}\right\rangle $ is maximally
entangled, this protocol is guaranteed to achieve the global optimum
for any set of states. }. Observe that the presence of $\left|\tau_{\varepsilon}\right\rangle $
enhances the local success probability; in particular, the more entangled
$\left|\tau_{\varepsilon}\right\rangle $ is, the higher is the local
success probability. But perfect discrimination is possible if and
only if $\varepsilon=0$, i.e., when $\left|\tau_{\varepsilon}\right\rangle $
is maximally entangled. 

\subsection*{Optimal resource states and entanglement cost }

Suppose the states of a given bipartite or multipartite ensemble $\mathcal{E}=\left\{ \left(p_{i},\rho_{i}\right):i=1,\dots,N\right\} $
cannot be optimally discriminated by LOCC. Let $\tau=\vert\tau\rangle\langle\tau\vert$
be a bipartite or multipartite ancilla state such that 
\begin{equation}
p_{_{{\rm L}}}\left(\mathcal{E}\otimes\tau\right)>p_{_{{\rm L}}}\left(\mathcal{E}\right),\label{useful resource}
\end{equation}
where $\mathcal{E}\otimes\tau=\left\{ \left(p_{i},\rho_{i}\otimes\tau\right):i=1,\dots,N\right\} $.
We then say that $\left|\tau\right\rangle $ is a resource for discriminating
the states of $\mathcal{E}$. 

Ultimately, however, the goal is to find a $\left|\tau\right\rangle $
that enables optimal discrimination of $\mathcal{E}$ by LOCC and
is also minimal in entanglement. The latter condition is imposed because
entanglement is generally regarded as an $expensive$ resource, and
therefore, we would like to consume as little entanglement as possible.
Satisfying the first condition is easy because one can always use
maximally entangled pair(s) and employ the teleportation protocol.
For example, any set from $\mathbb{C}^{d}\otimes\mathbb{C}^{d}$ (see
\cite{BHN-2016,BHN-2018} for discussions regarding multipartite systems)
can be optimally discriminated using LOCC and a $\mathbb{C}^{d}\otimes\mathbb{C}^{d}$
maximally entangled state as a resource. But finding a $\left|\tau\right\rangle $
that not only discriminates the states optimally but is also minimal
in entanglement is hard. That is because the teleportation protocol
using maximally entangled state(s) may not be the most efficient strategy
all the time, as one might do just as well with a clever protocol
that consumes less entanglement. 

For a given ensemble $\mathcal{E}$, we say that $\left|\tau\right\rangle \in\mathcal{H}_{\tau}$
is an $optimal$ resource if it enables optimal discrimination of
the states of $\mathcal{E}$ by LOCC, i.e., $p_{_{{\rm L}}}\left(\mathcal{E}\otimes\tau\right)=p\left(\mathcal{E}\right)$,
and is minimal in both entanglement and dimension \cite{BHN-2018},
i.e., for any other $\left|\tau^{\prime}\right\rangle \in\mathcal{H}_{\tau^{\prime}}$
satisfying $p_{_{{\rm L}}}\left(\mathcal{E}\otimes\tau^{\prime}\right)=p\left(\mathcal{E}\right)$,
it holds that $E\left(\tau\right)\leq E\left(\tau^{\prime}\right)$,
where $E$ is the entanglement entropy \cite{Entanglement-horodecki},
and $\dim\left(\mathcal{H}_{\tau}\right)\leq\dim\left(\mathcal{H}_{\tau^{\prime}}\right)$.
The entanglement of an optimal resource is said to be the $entanglement$
$cost$ of discriminating the states under consideration. For example,
a maximally entangled state is optimal for discriminating a maximally
entangled basis on $\mathbb{C}^{d}\otimes\mathbb{C}^{d}$ by LOCC\footnote{This follows from an argument that perfect discrimination of such
a basis would lead to distillation of a $\mathbb{C}^{d}\otimes\mathbb{C}^{d}$
maximally entangled state across a bipartition where there was no
entanglement to begin with \cite{Ghosh-2001,Ghosh-2004}.}, so the entanglement cost here is $\log_{2}d$ ebits. However, a
maximally entangled state is not always an optimal resource: a two-qubit
ensemble consisting of eight pure entangled states can be optimally
discriminated by LOCC using a nonmaximally entangled state \cite{BGKW-2009}. 

\subsection*{Motivation}

In entanglement-assisted local state discrimination, we are mainly
interested in finding the entanglement cost of discriminating the
states (or simply, entanglement cost) of some given ensemble. So which
factors determine the entanglement cost for a given ensemble? For
a bipartite orthonormal basis, the average entanglement of the basis
vectors provides a lower bound \cite{BGKW-2009,BRW-2010}. This lower
bound can be improved upon for almost all two-qubit entangled bases
and was shown to be strictly greater than the average entanglement
of the basis vectors \cite{BRW-2010}; however, finding the exact
values remains an open problem. 

The exact entanglement cost, however, is known for a few ensembles.
The entanglement cost is 1 ebit for the Bell basis \cite{Ghosh-2001},
a set of three Bell states \cite{B-IQC-2015}, and a set containing
a Bell state plus its orthogonal complement \cite{Yu-duan-PPT-2014},
but is less than an ebit for a set of nonorthogonal pure states and
is given by their average entanglement \cite{BGKW-2009}. In higher
dimensions, to the best of our knowledge, exact results are known
for a maximally entangled basis and a set containing a maximally entangled
state and its orthogonal complement \cite{Yu-duan-PPT-2014}.

Generally speaking, for an arbitrary ensemble $\mathcal{E}$, finding
the entanglement cost of LOCC discrimination, which is equivalent
to the problem of finding an optimal resource state, seems quite hard,
even if the states are from $\mathbb{C}^{2}\otimes\mathbb{C}^{2}$,
the smallest composite state space. Alternatively, one might ask:
For which bipartite ensembles is a maximally entangled state necessary
for optimal discrimination by LOCC? Now the ensembles from $\mathbb{C}^{2}\otimes\mathbb{C}^{2}$
that are known to require a Bell state have a common feature: Each
contains at least one of the four Bell states. This, therefore, raises
a very basic question: Suppose that we are given a set of two-qubit
states that does not contain a maximally entangled state and that
cannot be optimally discriminated by LOCC. Do we still require a Bell
state for optimal discrimination by LOCC? The present paper answers
this question in the affirmative and also shows that a Bell state
is required for optimal discrimination of some sets that do not even
contain an entangled state. 

\subsection*{Problem statement and overview of results }

Specifically, we study the problem of discriminating a set of ``noisy''
Bell states by LOCC assuming a uniform probability distribution; that
is, each state has an equal chance of being distributed to Alice and
Bob. A noisy Bell state, in general, results from actions of quantum
channels on one or both qubits of the Bell state in question. In fact,
the task of LOCC discrimination of four Bell states in a realistic
scenario boils down to LOCC discrimination of four noisy Bell states.
That is because the unknown Bell state must be distributed to Alice
and Bob through quantum channels that are noisy in practice. 

In this paper we shall assume that a noisy Bell state results from
mixing a Bell state with a two-qubit state with probabilities $\lambda$
and $\left(1-\lambda\right)$, where $0\leq\lambda\leq1$, or as a
consequence of the action of a quantum channel that leaves a Bell
state unchanged with probability $\lambda$ but converts it into a
two-qubit state with probability $\left(1-\lambda\right)$. 

Let $\mbox{D}\left(\mathbb{C}^{2}\otimes\mathbb{C}^{2}\right)$ be
the set of all two-qubit density matrices. Let $\Psi_{i}=\left|\Psi_{i}\right\rangle \left\langle \Psi_{i}\right|$
denote the density operator corresponding to the Bell state $\left|\Psi_{i}\right\rangle $
given by (\ref{Bell-basis}), and let $\varsigma$ be the density
operator corresponding to a two-qubit state that could be either pure
or mixed.

Consider a uniform collection of noisy Bell states 
\begin{equation}
\mathcal{B}_{\lambda,\varsigma}=\left\{ \varrho_{i}:i=1,\dots,4\right\} \subset\mbox{D}\left(\mathbb{C}^{2}\otimes\mathbb{C}^{2}\right),\label{set-S}
\end{equation}
where 
\begin{equation}
\varrho_{i}=\lambda\Psi_{i}+\left(1-\lambda\right)\varsigma\label{noisy Bell state}
\end{equation}
for $\lambda\in\left[0,1\right]$. The set $\mathcal{B}_{\lambda,\varsigma}$
is therefore completely determined by both $\lambda$ and $\varsigma$.
Of course, the situation where $\lambda=0$ is not interesting. 

Observe that for a fixed $j\in\left\{ 1,\dots,4\right\} $ one has
\begin{equation}
\max_{i\in\left\{ 1,\dots,4\right\} }\left\langle \Psi_{j}\right|\varrho_{i}\left|\Psi_{j}\right\rangle =\left\langle \Psi_{j}\right|\varrho_{j}\left|\Psi_{j}\right\rangle .\label{imp-eq.}
\end{equation}
Equation (\ref{imp-eq.}) means the following: Suppose a Bell measurement
is performed on a two-qubit system that has been prepared with equal
probability in one of $\left\{ \varrho_{i}\right\} $. Then, given
an outcome $j$, where $j\in\left\{ 1,\dots,4\right\} $, the system
was most likely prepared in the state $\varrho_{j}$. Note that, in
general,
\begin{equation}
\max_{j\in\left\{ 1,\dots,4\right\} }\left\langle \Psi_{j}\right|\varrho_{i}\left|\Psi_{j}\right\rangle \neq\left\langle \Psi_{i}\right|\varrho_{i}\left|\Psi_{i}\right\rangle .\label{imp-eq.2}
\end{equation}
The distinction between (\ref{imp-eq.}) and (\ref{imp-eq.2}) is
important. 

How well the states $\varrho_{i}$ can be discriminated is quantified
by the global optimum $p\left(\mathcal{B}_{\lambda,\varsigma}\right)$.
As the states are nonorthogonal except for $\lambda=1$, it holds
that $p\left(\mathcal{B}_{\lambda,\varsigma}\right)\leq1$ where equality
holds if and only if $\lambda=1$, i.e., for the Bell basis. Our first
result is a lower bound on $p\left(\mathcal{B}_{\lambda,\varsigma}\right)$:
For $\lambda\in\left[0,1\right]$ and any two-qubit state $\varsigma$,
it holds that
\begin{equation}
p\left(\mathcal{B}_{\lambda,\varsigma}\right)\geq\frac{1}{4}\left(1+3\lambda\right).\label{lower-bound-global}
\end{equation}
Later we will find that the lower bound is, in fact, the exact formula. 

Next, we compute the local optimum. The local success probability
of discriminating the states of $\mathcal{B}_{\lambda,\varsigma}$
is: 
\begin{equation}
p_{_{{\rm L}}}\left(\mathcal{B}_{\lambda,\varsigma}\right)=\frac{1}{4}\left(1+\lambda\right)\label{local optimum}
\end{equation}
for $\lambda\in\left[0,1\right]$ and any two-qubit state $\varsigma$.
Now observe that 
\[
p_{_{{\rm L}}}\left(\mathcal{B}_{\lambda,\varsigma}\right)<\frac{1}{4}\left(1+3\lambda\right)\leq p\left(\mathcal{B}_{\lambda,\varsigma}\right)\;\;\forall\lambda\in\left(0,1\right],
\]
which, in turn, proves that the states of $\mathcal{B}_{\lambda,\varsigma}$
cannot be optimally discriminated by LOCC for any $\lambda\in\left(0,1\right]$
and any two-qubit $\varsigma$. 

So the next thing is to find the entanglement cost of discriminating
the states of $\mathcal{B}_{\lambda,\varsigma}$ using LOCC. We assume
that $\left|\tau_{\varepsilon}\right\rangle $ $\left[\mbox{given by }\eqref{tau}\right]$
is used as a resource. 

First, we obtain the success probability of discriminating the states
of $\mathcal{B}_{\lambda,\varsigma}$ using LOCC and $\left|\tau_{\varepsilon}\right\rangle $.
The local success probability of distinguishing the states of 
\[
\mathcal{B}_{\lambda,\varsigma}\otimes\tau_{\varepsilon}=\left\{ \varrho_{i}\otimes\tau_{\varepsilon}:i=1,\dots,4\right\} 
\]
is given by: 
\begin{equation}
p_{_{{\rm L}}}\left(\mathcal{B}_{\lambda,\varsigma}\otimes\tau_{\varepsilon}\right)=\frac{1}{4}\left(1+\lambda+2\lambda\sqrt{1-\varepsilon^{2}}\right)\label{optimal-prob-ELOCC}
\end{equation}
for $\varepsilon\in\left[0,1\right]$, $\lambda\in\left[0,1\right]$,
and any two-qubit state $\varsigma$. Equation (\ref{optimal-prob-ELOCC})
can also be written as
\[
p_{_{{\rm L}}}\left(\mathcal{B}_{\lambda,\varsigma}\otimes\tau_{\varepsilon}\right)=p_{_{{\rm L}}}\left(\mathcal{B}_{\lambda,\varsigma}\right)+\frac{1}{2}\lambda\sqrt{1-\varepsilon^{2}}.
\]
Observe the contribution of the resource in the above equation, which
is given by the the second term on the right hand side for all $\varepsilon\in\left[0,1\right)$.
In particular, the presence of $\left|\tau_{\varepsilon}\right\rangle $
for any $\varepsilon\in\left[0,1\right)$ enhances the ability to
discriminate the states of $\mathcal{B}_{\lambda,\varsigma}$ by LOCC. 

Equation (\ref{optimal-prob-ELOCC}) leads to the formula for the
global optimum $p\left(\mathcal{B}_{\lambda,\varsigma}\right)$: 
\begin{equation}
p\left(\mathcal{B}_{\lambda,\varsigma}\right)=\frac{1}{4}\left(1+3\lambda\right)\label{global-optimum}
\end{equation}
for $\lambda\in\left[0,1\right]$ and any two-qubit state $\varsigma$.
So the lower bound from (\ref{lower-bound-global}) turns out to be
exact. 

Now, for $\lambda\in\left(0,1\right]$, 
\begin{equation}
p_{_{{\rm L}}}\left(\mathcal{B}_{\lambda,\varsigma}\otimes\tau_{\varepsilon}\right)\leq p\left(\mathcal{B}_{\lambda,\varsigma}\right),\label{cost-equality}
\end{equation}
where the equality holds if and only if $\varepsilon=0$. This, therefore,
gives us the entanglement cost. In particular, the entanglement cost
of discriminating the states of $\mathcal{B}_{\lambda,\varsigma}$
by LOCC is $1$ ebit for any $\lambda\in\left(0,1\right]$ and two-qubit
state $\varsigma$. Thus, a maximally entangled state is required
for optimal discrimination of the states of $\mathcal{B}_{\lambda,\varsigma}$
by LOCC, although for any given value of $\lambda\in\left(0,1\right)$
the ensemble, in general, does not contain any maximally entangled
state. This shows that a maximally entangled state may be required
to optimally discriminate a set of states none of which are maximally
entangled. 
\begin{rem}
The success probabilities $p_{_{{\rm L}}}\left(\mathcal{B}_{\lambda,\varsigma}\right)$,
$p_{_{{\rm L}}}\left(\mathcal{B}_{\lambda,\varsigma}\otimes\tau_{\varepsilon}\right)$,
and $p\left(\mathcal{B}_{\lambda,\varsigma}\right)$ are all independent
of $\varsigma$. This is not something that was expected \emph{a priori}
but seems to be the consequence of the fact that the Bell states are
all mixed with the same $\varsigma$. We should not expect something
similar if different two-qubit states are mixed with different Bell
states. 

The entanglement cost is seen to be independent of the gap between
the local and global optima that are given by (\ref{local optimum})
and (\ref{global-optimum}), respectively. As long as the gap remains
finite, no matter how small, the entanglement cost remains $1$ ebit,
irrespective of the entanglement or other properties of the states. 
\end{rem}

We illustrate the results with an example in which $\varsigma$ is
taken to be the maximally mixed state $\frac{1}{4}\boldsymbol{\mathsf{1}}$,
where $\boldsymbol{\mathsf{1}}$ is the identity operator acting on
the two-qubit state space. The general result tells us that the entanglement
cost is $1$ ebit for $\lambda\in\left(0,1\right]$. In this case,
the states are entangled for $\lambda\in\left(\frac{1}{3},1\right]$
but separable for $\lambda\in\left(0,\frac{1}{3}\right]$. So if we
consider a set $\mathcal{B}_{\lambda,\varsigma}$ for some $\lambda\in\left(0,\frac{1}{3}\right]$
and $\varsigma=\frac{1}{4}\boldsymbol{\mathsf{1}}$, then such a set
contains only separable states. Nevertheless, optimal discrimination
by LOCC requires a two-qubit maximally entangled state as a resource. 

\section{Preliminaries }

There is no tractable characterization of the set of LOCC measurements.
In fact, even deciding whether a measurement on a composite system
describes an LOCC measurement is computationally hard. For these reasons,
LOCC state discrimination problems are often investigated by considering
the more tractable classes: separable (SEP) measurements \cite{Duan-2009,BN-2013,B-IQC-2015}
and positive partial transpose (PPT) measurements \cite{Cosentino-2013,Cosentino-Russo-2014,Yu-duan-PPT-2014,B-IQC-2015}.
A separable measurement is one in which the measurement operators
are all separable, and a PPT measurement is one in which the measurement
operators are all positive under partial transposition. These measurements
often yield useful results and insights. One accordingly defines $p_{_{{\rm SEP}}}\left(\mathcal{E}\right)$
as the separable success probability and $p_{_{{\rm PPT}}}\left(\mathcal{E}\right)$
as the PPT success probability. Since
\begin{equation}
\left\{ \mbox{LOCC}\right\} \subset\left\{ \mbox{SEP}\right\} \subset\left\{ \mbox{PPT}\right\} \subset\left\{ \mbox{all}\right\} ,\label{inclusions}
\end{equation}
it holds that
\begin{equation}
p_{_{{\rm L}}}\left(\mathcal{E}\right)\leq p_{_{{\rm SEP}}}\left(\mathcal{E}\right)\leq p_{_{{\rm PPT}}}\left(\mathcal{E}\right)\leq p\left(\mathcal{E}\right).\label{series inequalities}
\end{equation}
Note that, if $\mathcal{E}$ can be optimally discriminated by LOCC,
then the above inequalities turn into equalities. On the other hand,
if $\mathcal{E}$ can be optimally discriminated by a separable measurement
but not by LOCC \cite{ben99,ben99u}, then only the first inequality
is strict. It may also be instructive to note that in the case of
four Bell states only the last one is strict. 

Let $\mathcal{X}$ and $\mathcal{Y}$ represent finite-dimensional
Hilbert spaces associated with quantum systems that belong to Alice
and Bob respectively. Let $\text{Pos}\left(\mathcal{X}\right)$, $\text{Pos}\left(\mathcal{Y}\right)$,
and $\text{Pos}\left(\mathcal{X}\otimes\mathcal{Y}\right)$ denote
the sets of positive semidefinite operators acting on $\mathcal{X}$,
$\mathcal{Y}$, and $\mathcal{X}\otimes\mathcal{Y}$, respectively.
An operator $P\in\text{Pos}\left(\mathcal{X}\otimes\mathcal{Y}\right)$
is PPT if ${\rm T}_{\mathcal{X}}\left(P\right)\in\text{Pos}\left(\mathcal{X}\otimes\mathcal{Y}\right)$,
where $\text{T}_{\mathcal{X}}$ represents partial transposition taken
in the standard basis $\left\{ \left|0\right\rangle ,\dots,\left|d-1\right\rangle \right\} $
of $\mathcal{X}$. A PPT measurement is defined by a collection of
measurement operators $\left\{ P_{1},\dots,P_{N}\right\} $ in which
each operator is PPT. 

Let us denote the set of all PPT operators acting on $\mathcal{X}\otimes\mathcal{Y}$
by $\text{PPT}\left(\mathcal{X}:\mathcal{Y}\right)$. The set $\text{PPT}\left(\mathcal{X}:\mathcal{Y}\right)$
is a closed, convex cone. For a given ensemble $\mathcal{E}=\left\{ \left(p_{1},\rho_{1}\right),\dots,\left(p_{N},\rho_{N}\right)\right\} \subset\mathcal{X}\otimes\mathcal{Y}$
the problem of finding $p_{_{{\rm PPT}}}\left(\mathcal{E}\right)$
can be expressed as a semidefinite program \cite{Cosentino-2013}:\\

$\mathtt{Primal}\;\mathtt{problem:}$

\begin{eqnarray*}
\mathtt{maximize}:\;\; &  & \sum_{i=1}^{N}p_{i}\text{Tr}\left(\rho_{i}P_{i}\right)\\
\mathtt{subject\,\mathtt{to}}: &  & \sum_{i=1}^{N}P_{i}=\boldsymbol{\mathsf{1}}_{\mathcal{X}\otimes\mathcal{Y}}\\
 &  & P_{k}\in\text{PPT}\left(\mathcal{X}:\mathcal{Y}\right)\;\;\left(\text{for\;\ each }k=1,\dots,N\right)
\end{eqnarray*}

$\mathtt{Dual\;problem:}$
\begin{eqnarray*}
\mathtt{minimize}:\;\; &  & \text{Tr}\left(H\right)\\
\mathtt{subject}\,\mathtt{to}: &  & H-p_{k}\rho_{k}\in\text{PPT}\left(\mathcal{X}:\mathcal{Y}\right)\;\;\left(\text{for\;\ each }k=1,\dots,N\right)\\
 &  & H\in\text{Herm}\left(\mathcal{X}\otimes\mathcal{Y}\right),
\end{eqnarray*}
where $\text{Herm}\left(\mathcal{X}\otimes\mathcal{Y}\right)$ is
the set of Hermitian operators acting on $\mathcal{X}\otimes\mathcal{Y}$.
By weak duality every feasible solution of the dual problem provides
an upper bound on $p_{_{{\rm PPT}}}\left(\mathcal{E}\right)$.

\section{LOCC discrimination of $\mathcal{B}_{\lambda,\varsigma}$}

In this section we prove that the states of a set $\mathcal{B}_{\lambda,\varsigma}$
cannot be optimally discriminated by LOCC for $\lambda\in\left(0,1\right]$
and any choice of a two-qubit state $\varsigma$. 

Let $\mathcal{X}_{1}=\mathbb{C}^{2}$ and $\mathcal{Y}_{1}=\mathbb{C}^{2}$
denote the Hilbert spaces of Alice and Bob respectively. First, we
give a lower bound on $p\left(\mathcal{B}_{\lambda,\varsigma}\right)$. 
\begin{lem}
\label{lower-bound on p(S)} $p\left(\mathcal{B}_{\lambda,\varsigma}\right)\geq\frac{1}{4}\left(1+3\lambda\right)$
for $\lambda\in\left[0,1\right]$ and any $\varsigma\in\mbox{D}\left(\mathcal{X}_{1}\otimes\mathcal{Y}_{1}\right)$. 
\end{lem}

\begin{proof}
For any quantum measurement $\left\{ M_{a}\right\} $ on $\mathcal{X}_{1}\otimes\mathcal{Y}_{1}$,
it holds that
\[
p\left(\mathcal{B}_{\lambda,\varsigma}\right)\geq\frac{1}{4}\sum_{a}\max_{i\in\left\{ 1,\dots,4\right\} }\mbox{Tr}\left(\varrho_{i}M_{a}\right).
\]
 Choosing $\left\{ M_{a}\right\} $ to be the Bell measurement $\left\{ \Psi_{1},\Psi_{2},\Psi_{3},\Psi_{4}\right\} $,
we get 
\begin{eqnarray}
p\left(\mathcal{B}_{\lambda,\varsigma}\right) & \geq & \frac{1}{4}\sum_{a=1}^{4}\max_{i\in\left\{ 1,\dots,4\right\} }\left\langle \Psi_{a}\left|\varrho_{i}\right|\Psi_{a}\right\rangle .\label{lemma-2-second}
\end{eqnarray}
Noting that $\max_{i\in\left\{ 1,\dots,4\right\} }\left\langle \Psi_{a}\left|\varrho_{i}\right|\Psi_{a}\right\rangle =\left\langle \Psi_{a}\left|\varrho_{a}\right|\Psi_{a}\right\rangle $,
we can write (\ref{lemma-2-second}) as
\begin{alignat}{1}
p\left(\mathcal{B}_{\lambda,\varsigma}\right) & \geq\frac{1}{4}\sum_{a=1}^{4}\left\langle \Psi_{a}\left|\varrho_{a}\right|\Psi_{a}\right\rangle \nonumber \\
 & =\lambda+\left(\frac{1-\lambda}{4}\right)\sum_{a=1}^{4}\left\langle \Psi_{a}\left|\varsigma\right|\Psi_{a}\right\rangle \nonumber \\
 & =\frac{1}{4}\left(1+3\lambda\right).\label{lemma2-last equation}
\end{alignat}
To arrive at the last line we have used $\sum_{a=1}^{4}\left\langle \Psi_{a}\left|\varsigma\right|\Psi_{a}\right\rangle =1$.
Clearly, (\ref{lemma2-last equation}) holds for all $\lambda\in\left[0,1\right]$
and any $\varsigma$. 
\end{proof}
\begin{lem}
\label{p-local(S)} $p_{_{{\rm L}}}\left(\mathcal{B}_{\lambda,\varsigma}\right)=\frac{1}{4}\left(1+\lambda\right)$
for $\lambda\in\left[0,1\right]$ and any $\varsigma\in\mbox{D}\left(\mathbb{\mathcal{X}}_{1}\otimes\mathcal{Y}_{1}\right)$. 
\end{lem}

\begin{proof}
The proof contains two parts. First, we show that $p_{_{{\rm PPT}}}\left(\mathcal{B}_{\lambda,\varsigma}\right)\leq\frac{1}{4}\left(1+\lambda\right)$
and then we will give a local protocol that achieves this bound. 

Let $\lambda\in\left[0,1\right]$. Consider the operator
\begin{equation}
H_{\lambda}=\frac{1}{8}\left[\lambda\boldsymbol{\mathsf{1}}_{\mathcal{X}_{1}\otimes\mathcal{Y}_{1}}+2\left(1-\lambda\right)\varsigma\right]\in\mbox{Herm}\left(\mathcal{X}_{1}\otimes\mathcal{Y}_{1}\right),\label{H-lambda}
\end{equation}
where $\boldsymbol{\mathsf{1}}_{\mathcal{X}_{1}\otimes\mathcal{Y}_{1}}$
is the identity operator acting on $\mathcal{X}_{1}\otimes\mathcal{Y}_{1}$.
Then
\begin{equation}
\mbox{Tr}\left(H_{\lambda}\right)=\frac{1}{4}\left(1+\lambda\right).\label{Trace-H-Lambda}
\end{equation}
We will now show that $H_{\lambda}$ is a feasible solution of the
dual of the PPT state discrimination problem. In particular, we will
show that
\begin{equation}
\mbox{T}_{\mathcal{X}_{1}}\left(H_{\lambda}-\frac{1}{4}\varrho_{i}\right)\in\mbox{Pos}\left(\mathcal{X}_{1}\otimes\mathcal{Y}_{1}\right)\;\forall i=1,\dots,4,\label{dual-feasibility-1}
\end{equation}
which is a sufficient condition for dual feasibility. 

Observe that 
\begin{alignat*}{1}
H_{\lambda}-\frac{1}{4}\varrho_{i} & =\frac{1}{8}\left(\lambda\boldsymbol{\mathsf{1}}_{\mathcal{X}_{1}\otimes\mathcal{Y}_{1}}-2\lambda\Psi_{i}\right)\\
 & =\frac{\lambda}{4}\left(\frac{1}{2}\boldsymbol{\mathsf{1}}_{\mathcal{X}_{1}\otimes\mathcal{Y}_{1}}-\Psi_{i}\right)\\
 & =\frac{\lambda}{4}\mbox{T}_{\mathcal{X}_{1}}\left(\Psi_{5-i}\right)\;\;\hspace{1em}\hspace{1em}\left(\mbox{for every }i=1,\dots,4\right)
\end{alignat*}
Hence 
\[
\mbox{T}_{\mathcal{X}_{1}}\left(H_{\lambda}-\frac{1}{4}\varrho_{i}\right)=\frac{\lambda}{4}\Psi_{i-5}\in\mbox{Pos}\left(\mathcal{X}_{1}\otimes\mathcal{Y}_{1}\right)
\]
for every $i=1,\dots,4$. So by weak duality we have
\begin{equation}
p_{_{{\rm PPT}}}\left(\mathcal{B}_{\lambda,\varsigma}\right)\leq\mbox{Tr}\left(H_{\lambda}\right)=\frac{1}{4}\left(1+\lambda\right).\label{upper-bound}
\end{equation}
Consequently, 
\begin{equation}
p_{_{{\rm L}}}\left(\mathcal{B}_{\lambda,\varsigma}\right)\leq p_{_{{\rm PPT}}}\left(\mathcal{B}_{\lambda,\varsigma}\right)\leq\frac{1}{4}\left(1+\lambda\right).\label{upperbound}
\end{equation}
We will now show that the upper bound (\ref{upperbound}) is also
a lower bound on $p_{_{{\rm L}}}\left(\mathcal{B}_{\lambda,\varsigma}\right)$.
Choosing the local measurement in the computational basis $\left\{ \left|a\right\rangle :a\in\left\{ 00,01,10,11\right\} \right\} $,
we get
\begin{alignat}{1}
p_{_{{\rm L}}}\left(\mathcal{B}_{\lambda,\varsigma}\right) & \geq\frac{1}{4}\sum_{a}\max_{i}\left\langle a\left|\varrho_{i}\right|a\right\rangle \nonumber \\
 & =\frac{\lambda}{2}+\left(\frac{1-\lambda}{4}\right)\sum_{a}\left\langle a\left|\varsigma\right|a\right\rangle \nonumber \\
 & =\frac{1}{4}\left(1+\lambda\right).\label{lower-bound}
\end{alignat}
From (\ref{upperbound}) and (\ref{lower-bound}) it follows that
\begin{equation}
p_{_{{\rm L}}}\left(\mathcal{B}_{\lambda,\varsigma}\right)=\frac{1}{4}\left(1+\lambda\right)\label{p-L-S}
\end{equation}
for $\lambda\in\left[0,1\right]$ and any two-qubit state $\varsigma$. 
\end{proof}
Lemmas \ref{lower-bound on p(S)} and \ref{p-local(S)} together imply:
\[
p_{_{{\rm L}}}\left(\mathcal{B}_{\lambda,\varsigma}\right)<\frac{1}{4}\left(1+3\lambda\right)\leq p\left(\mathcal{B}_{\lambda,\varsigma}\right)\;\;\mbox{for }\lambda\in\left(0,1\right],
\]
which proves the following theorem. 
\begin{thm}
The states of a set $\mathcal{B}_{\lambda,\varsigma}$, as defined
by (\ref{set-S}), cannot be optimally discriminated by LOCC for any
$\lambda\in\left(0,1\right]$ and any two-qubit state $\varsigma$. 
\end{thm}

In the next section, we take up the question of finding the entanglement
cost: the amount of entanglement one must consume to optimally discriminate
the states of a set $\mathcal{B}_{\lambda,\varsigma}$, where $\lambda\in\left(0,1\right]$,
by LOCC. 

\section{The entanglement cost of discriminating $\mathcal{B}_{\lambda,\varsigma}$ }

Let us now assume that Alice and Bob share an additional resource
state $\left|\tau_{\varepsilon}\right\rangle \in\mathcal{X}_{2}\otimes\mathcal{Y}_{2}$
defined by (\ref{tau}), where $\mathcal{X}_{2}=\mathbb{C}^{2}$ and
$\mathcal{Y}_{2}=\mathbb{C}^{2}$ are the Hilbert spaces associated
with the ancilla systems. That means we now consider the task of LOCC
discrimination of the states corresponding to the set
\begin{equation}
\mathcal{B}_{\lambda,\varsigma}\otimes\tau_{\varepsilon}=\left\{ \varrho_{i}\otimes\tau_{\varepsilon}:i=1,\dots,4\right\} \subset\left(\mathcal{X}_{1}\otimes\mathcal{Y}_{1}\right)\otimes\left(\mathcal{X}_{2}\otimes\mathcal{Y}_{2}\right),\label{S cross tau}
\end{equation}
where the states are all equally probable, and $\tau_{\varepsilon}=\left|\tau_{\varepsilon}\right\rangle \left\langle \tau_{\varepsilon}\right|\subset\mbox{D}\left(\mathcal{X}_{2}\otimes\mathcal{Y}_{2}\right)$. 
\begin{thm}
The local success probability of discriminating the states of $\mathcal{B}_{\lambda,\varsigma}\otimes\tau_{\varepsilon}$
is given by 
\begin{equation}
p_{_{L}}\left(\mathcal{B}_{\lambda,\varsigma}\otimes\tau_{\varepsilon}\right)=\frac{1}{4}\left(1+\lambda+2\lambda\sqrt{1-\varepsilon^{2}}\right)\label{Theorem-5}
\end{equation}
for $\varepsilon\in\left[0,1\right]$, $\lambda\in\left[0,1\right]$,
and any $\varsigma\in\mbox{D}\left(\mathcal{X}_{1}\otimes\mathcal{Y}_{1}\right)$. 
\end{thm}

\begin{proof}
Let $\varepsilon\in\left[0,1\right]$, $\lambda\in\left[0,1\right]$,
and $\varsigma\in\mbox{D}\left(\mathcal{X}_{1}\otimes\mathcal{Y}_{1}\right)$.
First, we will prove that 
\begin{equation}
p_{_{{\rm PPT}}}\left(\mathcal{B}_{\lambda,\varsigma}\otimes\tau_{\varepsilon}\right)\leq\frac{1}{4}\left(1+\lambda+2\lambda\sqrt{1-\varepsilon^{2}}\right),\label{eupper-bound-p-PPT}
\end{equation}
and then we will give a local protocol that achieves this upper bound. 

Define the operator: 
\begin{equation}
H_{\lambda,\varepsilon}=\lambda H_{\varepsilon}+\left(\frac{1-\lambda}{4}\right)\varsigma\otimes\tau_{\varepsilon}\in\mbox{Herm}\left(\mathcal{X}_{1}\otimes\mathcal{Y}_{1}\otimes\mathcal{X}_{2}\otimes\mathcal{Y}_{2}\right),\label{H-lambda-eps}
\end{equation}
where 
\begin{equation}
H_{\epsilon}=\frac{1}{8}\left[\boldsymbol{\mathsf{1}}_{\mathcal{X}_{1}\otimes\mathcal{Y}_{1}}\otimes\tau_{\varepsilon}+\sqrt{1-\varepsilon^{2}}\mathsf{\boldsymbol{1}}_{\mathcal{X}_{1}\otimes\mathcal{Y}_{1}}\otimes\mbox{T}_{\mathcal{X}_{2}}\left(\Psi_{4}\right)\right]\in\mbox{Herm}\left(\mathcal{X}_{1}\otimes\mathcal{Y}_{1}\otimes\mathcal{X}_{2}\otimes\mathcal{Y}_{2}\right).\label{H-epsilon}
\end{equation}
It holds that
\begin{equation}
\mbox{Tr}\left(H_{\lambda,\varepsilon}\right)=\frac{1}{4}\left(1+\lambda+2\lambda\sqrt{1-\varepsilon^{2}}\right).\label{Trace-H-Lambda-eps}
\end{equation}
We now show that $H_{\lambda,\varepsilon}$ is a feasible solution
of the dual problem of discriminating the states $\varrho_{i}\otimes\tau_{\varepsilon}$,
$i=1,\dots,4$, by PPT measurements. In particular, we will prove
that
\begin{equation}
\left(\mbox{T}_{\mathcal{X}_{1}}\otimes\mbox{T}_{\mathcal{X}_{2}}\right)\left(H_{\lambda,\varepsilon}-\frac{1}{4}\varrho_{i}\otimes\tau_{\varepsilon}\right)\in\mbox{Pos}\left(\mathcal{X}_{1}\otimes\mathcal{Y}_{1}\otimes\mathcal{X}_{2}\otimes\mathcal{Y}_{2}\right)\;\forall i=1,\dots,4,\label{dual-feasibility}
\end{equation}
which is a sufficient condition for dual feasibility. The proof is,
in fact, almost immediate. Observe that
\begin{alignat*}{1}
H_{\lambda,\varepsilon}-\frac{1}{4}\varrho_{i}\otimes\tau_{\varepsilon} & =\lambda\left(H_{\varepsilon}-\frac{1}{4}\Psi_{i}\otimes\tau_{\varepsilon}\right)\;\;\left(\mbox{for every }i=1,\dots,4\right).
\end{alignat*}
Therefore, 
\begin{alignat*}{1}
\left(\mbox{T}_{\mathcal{X}_{1}}\otimes\mbox{T}_{\mathcal{X}_{2}}\right)\left(H_{\lambda,\varepsilon}-\frac{1}{4}\varrho_{i}\otimes\tau_{\varepsilon}\right) & =\lambda\left(\mbox{T}_{\mathcal{X}_{1}}\otimes\mbox{T}_{\mathcal{X}_{2}}\right)\left(H_{\varepsilon}-\frac{1}{4}\Psi_{i}\otimes\tau_{\varepsilon}\right)\;\;\left(\mbox{for every }i=1,\dots,4\right)
\end{alignat*}
which is positive semidefinite \cite{B-IQC-2015}. So we have 
\[
\left(\mbox{T}_{\mathcal{X}_{1}}\otimes\mbox{T}_{\mathcal{X}_{2}}\right)\left(H_{\lambda,\varepsilon}-\frac{1}{4}\varrho_{i}\otimes\tau_{\varepsilon}\right)\in\mbox{Pos}\left(\mathcal{X}_{1}\otimes\mathcal{Y}_{1}\otimes\mathcal{X}_{2}\otimes\mathcal{Y}_{2}\right)
\]
for every $i=1,\dots,4$. 

By weak duality 
\begin{equation}
p_{_{{\rm PPT}}}\left(\mathcal{B}_{\lambda,\varsigma}\otimes\tau_{\varepsilon}\right)\leq\mbox{Tr}\left(H_{\lambda,\varepsilon}\right)=\frac{1}{4}\left(1+\lambda+2\lambda\sqrt{1-\varepsilon^{2}}\right).\label{upper bound on p(PPT)}
\end{equation}
Since $p_{_{L}}\left(\mathcal{B}_{\lambda,\varsigma}\otimes\tau_{\varepsilon}\right)\leq p_{_{{\rm PPT}}}\left(\mathcal{B}_{\lambda,\varsigma}\otimes\tau_{\varepsilon}\right)$,
it holds that 
\begin{equation}
p_{_{L}}\left(\mathcal{B}_{\lambda,\varsigma}\otimes\tau_{\varepsilon}\right)\leq\frac{1}{4}\left(1+\lambda+2\lambda\sqrt{1-\varepsilon^{2}}\right).\label{upperbound-p-l-s-tau}
\end{equation}
We now give a lower bound on $p_{_{{\rm L}}}\left(\mathcal{B}_{\lambda,\varsigma}\otimes\tau_{\varepsilon}\right)$.
The lower bound is obtained using a strategy based on the teleportation
protocol. First, Alice teleports her qubit to Bob using $\tau_{\varepsilon}$
as the teleportation channel following the standard protocol: Alice
performs the Bell measurement and informs Bob of the outcome, and
Bob then applies the relevant unitary operation\footnote{The convention is as follows: If Alice gets $\Psi_{1}$, Bob does
nothing; if Alice gets $\Psi_{2}$, Bob applies $\sigma_{z}$, etc. }. This results in Bob holding one of the four two-qubit states from
\begin{alignat*}{1}
\varrho_{1}^{\prime} & =\lambda\tau_{\varepsilon}+\left(1-\lambda\right)\varsigma^{\prime},\\
\varrho_{2}^{\prime} & =\lambda\left(\boldsymbol{\mathsf{1}}\otimes\sigma_{z}\right)\tau_{\varepsilon}\left(\boldsymbol{\mathsf{1}}\otimes\sigma_{z}\right)+\left(1-\lambda\right)\varsigma^{\prime},\\
\varrho_{3}^{\prime} & =\lambda\left(\boldsymbol{\mathsf{1}}\otimes\sigma_{x}\right)\tau_{\varepsilon}\left(\boldsymbol{\mathsf{1}}\otimes\sigma_{x}\right)+\left(1-\lambda\right)\varsigma^{\prime},\\
\varrho_{4}^{\prime} & =\lambda\left(\boldsymbol{\mathsf{1}}\otimes\sigma_{y}\right)\tau_{\varepsilon}\left(\boldsymbol{\mathsf{1}}\otimes\sigma_{y}\right)+\left(1-\lambda\right)\varsigma^{\prime},
\end{alignat*}
where $\sigma_{x},\sigma_{y},\sigma_{z}$ are the Pauli matrices and
$\varsigma^{\prime}$ is the post-teleportation $\varsigma$. Now
once the teleportation part is over, Bob performs a measurement to
discriminate the states $\varrho_{i}^{\prime}$. In particular, he
performs the Bell measurement $\left\{ \Psi_{1},\Psi_{2},\Psi_{3},\Psi_{4}\right\} $,
which leads to
\begin{alignat}{1}
p_{_{{\rm L}}}\left(\mathcal{B}_{\lambda,\varsigma}\otimes\tau_{\varepsilon}\right) & \geq\frac{1}{4}\sum_{a=1}^{4}\max_{i}\left\langle \Psi_{a}\left|\varrho_{i}^{\prime}\right|\Psi_{a}\right\rangle \nonumber \\
 & =\frac{\lambda}{2}\left(1+\sqrt{1-\varepsilon^{2}}\right)+\left(\frac{1-\lambda}{4}\right)\sum_{a=1}^{4}\left\langle \Psi_{a}\left|\varsigma^{\prime}\right|\Psi_{a}\right\rangle \nonumber \\
 & =\frac{1}{4}\left(1+\lambda+2\lambda\sqrt{1-\varepsilon^{2}}\right).\label{lower-bound-p-l-s-tau}
\end{alignat}
From (\ref{upperbound-p-l-s-tau}) and (\ref{lower-bound-p-l-s-tau})
we obtain the desired result. This completes the proof. 
\end{proof}
Now we see that 
\[
p_{_{{\rm L}}}\left(\mathcal{B}_{\lambda,\varsigma}\otimes\tau_{\varepsilon}\right)\leq\frac{1}{4}\left(1+3\lambda\right)\leq p\left(\mathcal{B}_{\lambda,\varsigma}\right)\;\text{for}\;\lambda\in\left(0,1\right],
\]
where the first inequality is an equality for $\varepsilon=0$. In
other words, the best possible local success probability is obtained
only when $\left|\tau_{\varepsilon}\right\rangle $ is maximally entangled,
and that must also be, in this case, the global optimum. So we have
\begin{equation}
p\left(\mathcal{B}_{\lambda,\varsigma}\right)=p_{_{{\rm L}}}\left(\mathcal{B}_{\lambda,\varsigma}\otimes\tau_{\varepsilon=0}\right)=\frac{1}{4}\left(1+3\lambda\right).\label{The-equality}
\end{equation}
Therefore, the lower bound in Lemma \ref{lower-bound on p(S)} is,
in fact, the global optimum. 

To summarize, we have proved that for any $\lambda\in\left(0,1\right]$
and any two-qubit state $\varsigma$
\begin{align*}
p_{_{{\rm L}}}\left(\mathcal{B}_{\lambda,\varsigma}\otimes\tau_{\varepsilon}\right) & <p\left(\mathcal{B}_{\lambda,\varsigma}\right)\;\text{for }\varepsilon\in\left(0,1\right]\\
p_{_{{\rm L}}}\left(\mathcal{B}_{\lambda,\varsigma}\otimes\tau_{\varepsilon}\right) & =p\left(\mathcal{B}_{\lambda,\varsigma}\right)\;\text{for }\varepsilon=0.
\end{align*}
So the states of $\mathcal{B}_{\lambda,\varsigma}$ for any $\lambda\in\left(0,1\right]$
and any two-qubit state $\varsigma$ can be optimally discriminated
by LOCC if and only if the resource state $\left|\tau_{\varepsilon}\right\rangle $
is maximally entangled, i.e., $\varepsilon=0$. Now an optimal resource
state is the one that enables optimal discrimination by LOCC and is
also minimal in both entanglement and dimension. Noting that $\left|\tau_{\varepsilon=0}\right\rangle $
is from $\mathbb{C}^{2}\otimes\mathbb{C}^{2}$, we conclude that it
is an optimal resource. Now recall that the entanglement cost of discriminating
a set of states by LOCC is given by the entanglement of an optimal
resource. We have the following theorem. 
\begin{thm}
\label{The-entanglement-cost} The entanglement cost of optimal discrimination
of the states of $\mathcal{B}_{\lambda,\varsigma}$ by LOCC is $1$
ebit for any $\lambda\in\left(0,1\right]$ and any two-qubit state
$\varsigma$. 
\end{thm}

\begin{rem}
Note that the entanglement cost in this case is independent of the
entanglement of the constituent states and also the choice of the
two-qubit state $\varsigma$. In fact, the entanglement cost is $1$
ebit as long as the gap between the local and global optima is nonzero.
Further, note that a set $\mathcal{B}_{\lambda,\varsigma}$ for $\lambda\in\left(0,1\right)$,
in general, does not contain a maximally entangled state. Such sets
are examples of sets that do not contain a maximally entangled state
but still require a maximally entangled state for optimal discrimination
by LOCC. 
\end{rem}

\begin{example}
\label{example} \emph{Bell states mixed with white noise}

Let us now consider a concrete example in which $\varsigma$ is taken
to be the maximally mixed state of two qubits, i.e., $\varsigma=\frac{1}{4}\boldsymbol{\mathsf{1}}_{\mathcal{X}_{1}\otimes\mathcal{Y}_{1}}$.
Then we have the following set of noisy Bell states: 
\[
\mathcal{B}_{\lambda,\frac{1}{4}\boldsymbol{\mathsf{1}}}=\left\{ \Omega_{i}:i=1,\dots,4\right\} ,
\]
where
\[
\Omega_{i}=\lambda\Psi_{i}+\frac{1-\lambda}{4}\boldsymbol{\mathsf{1}}_{\mathcal{X}_{1}\otimes\mathcal{Y}_{1}}
\]
for $\lambda\in\left(0,1\right]$. The results obtained earlier apply
straightaway. But now the range of $\lambda$ has a clear interpretation
in terms of the entanglement of the states: each state $\Omega_{i}$
, where $i=1,\dots,4$ is entangled if and only if $\lambda\in\left(\frac{1}{3},1\right]$.
So a set $\mathcal{B}_{\lambda,\frac{1}{4}\boldsymbol{\mathsf{1}}}$
contains entangled states for $\lambda\in\left(\frac{1}{3},1\right]$
and separable states for $\lambda\in\left(0,\frac{1}{3}\right]$.
But for any such set we now know that the entanglement cost of discrimination
by LOCC is $1$ ebit. So the entanglement cost, in this case, is independent
of the entanglement of the states. Furthermore, one requires a full
ebit even when the states are separable. 
\end{example}

\section{Conclusions}

A set of bipartite or multipartite quantum states cannot always be
optimally discriminated by LOCC. So given a set of states that cannot
be optimally discriminated by LOCC, a basic question is: How much
entanglement, as a resource, must one consume to perform the task
of optimal discrimination by LOCC? For instance, a set of three or
four Bell states can be perfectly discriminated by LOCC if and only
if a Bell state is used as a resource. 

In this paper we considered the problem of LOCC discrimination of
a uniform collection $\mathcal{B}_{\lambda,\varsigma}$ of noisy Bell
states that are obtained by mixing the Bell states with a two-qubit
state $\varsigma$ with probabilities $\lambda$ and $\left(1-\lambda\right)$.
First, we showed that the states of $\mathcal{B}_{\lambda,\varsigma}$
cannot be optimally discriminated by LOCC for any $\lambda\in\left(0,1\right]$
and $\varsigma$, so optimal discrimination will require an ancillary
entangled state. Since, such sets, in general, do not contain a maximally
entangled state, it was interesting to find out whether optimal discrimination
is possible without using a maximally entangled state. We, however,
proved that optimal discrimination by LOCC is possible if and only
if a two-qubit maximally entangled state is used as a resource for
any $\lambda\in\left(0,1\right]$ and $\varsigma$. So the result
holds regardless of the entanglement of the states, which could even
be separable in some cases. More specifically, the result holds as
long as the gap between the local and global optima is nonzero, no
matter how small. To prove our results we have utilized the fact that
determining the optimal value of discriminating via PPT measurements
can be represented as a semidefinite program \cite{Cosentino-2013}. 

There is at least one important application of results of this kind,
and that is related to the question of finding the entanglement cost
of nonlocal measurements \cite{BGKW-2009,BRW-2010}. For example,
the entanglement cost of local implementation of a nonlocal completely
orthogonal measurement must be at least as much as the entanglement
cost of locally discriminating the measurement eigenstates. That is
because if we could implement the measurement, we would be able to
perfectly distinguish the measurement eigenstates. This idea can be
extended to general nonlocal measurements as well. For instance, consider
the states in Example \ref{example}. It is straightforward to observe
that $\sum_{i=1}^{4}\Omega_{i}=\boldsymbol{\mathsf{1}}_{\mathcal{X}_{1}\otimes\mathcal{Y}_{1}}$.
Since $\Omega_{i}$ are positive operators, it follows that the collection
$\left\{ \Omega_{i}\right\} $ represents a noisy Bell measurement.
Our result, in particular, can be applied to obtain the entanglement
cost of locally implementing such measurements. 

An interesting open question is whether a nonmaximally entangled,
orthonormal basis of $\mathbb{C}^{2}\otimes\mathbb{C}^{2}$ can be
perfectly discriminated with LOCC using a nonmaximally entangled state.
One may, for instance, consider working with the bases in \cite{BGKW-2009,BRW-2010}
for which the lower bound was proved to be strictly larger than the
entropy bound given by the average entanglement of the states assuming
they are equally probable. 

Finally, we hope the results presented in this paper, particularly
the techniques \cite{Cosentino-2013,Cosentino-Russo-2014,B-IQC-2015}
used to prove the results, will be useful for future research in this
area.

\end{document}